\documentclass[10pt,a4paper]{article}

\makeatother
\usepackage[a4paper, total={6in, 8in}]{geometry}
\usepackage[makeroom]{cancel}
\usepackage{bbm}
\usepackage{tabu}
\usepackage{amsmath}
\usepackage{amsthm}
\usepackage{breqn}
\usepackage{verbatim}
\usepackage{graphicx}
\usepackage{caption}
\usepackage{subcaption}
\usepackage{esvect}
\usepackage{cite}
\usepackage{amssymb}
\usepackage{footnote}
\usepackage{color}
\usepackage{url}
\usepackage{xr}
\usepackage[makeroom]{cancel}
\theoremstyle{plain}
\newtheorem{theorem}{Theorem}[section]
\newtheorem{corollary}[theorem]{Corollary}
\newtheorem{definition}[theorem]{Definition}
\newtheorem{lemma}[theorem]{Lemma}
\newtheorem{proposition}[theorem]{Proposition}
\newtheorem{remark}{Remark}
\newtheorem*{problem statement}{Problem Statement}
\newtheorem{feasibility problem}{Problem}

\newtheorem{assumption}{Assumption}
\definecolor{myred}{RGB}{108,0,0}
\graphicspath{{./Figures/}}

\newcommand{\real}{\mathbb{R}}
\newcommand{\reg}{\mathcal{R}}
\newcommand{\setz}{\mathrm{Z}}
\newcommand{\setw}{\mathrm{W}}
\newcommand{\sete}{\mathrm{E}}
\newcommand{\init}{\mathcal{I}}
\newcommand{\reach}{\mathcal{X}}
\newcommand{\cone}{\mathcal{C}}

\begin{document}
\title{Region-Based Self-Triggered Control for Perturbed and Uncertain Nonlinear Systems}
\author{Giannis Delimpaltadakis and Manuel Mazo Jr.\thanks{The authors are with the Delft Center for Systems and Control, Delft university of Technology, Delft 2628CD, The Netherlands. Emails:\texttt{\{i.delimpaltadakis, m.mazo\}@tudelft.nl}. This work is supported by the ERC Starting Grant SENTIENT (755953).}
\date{}}
\maketitle
\graphicspath{{./Figures/}}
\begin{abstract}
	In this work, we derive a region-based self-triggered control (STC) scheme for nonlinear systems with bounded disturbances and model uncertainties. The proposed STC scheme is able to guarantee different performance specifications (e.g. stability, boundedness, etc.), depending on the event-triggered control (ETC) triggering function that is chosen to be emulated. To deal with disturbances and uncertainties, we employ differential inclusions (DIs). By introducing ETC/STC notions in the context of DIs, we extend well-known results on ETC/STC to perturbed uncertain systems. Given these results, and adapting tools from our previous work, we derive inner-approximations of isochronous manifolds of perturbed uncertain ETC systems. These approximations dictate a partition of the state-space into regions, each of which is associated to a uniform inter-sampling time. At each sampling time instant, the controller checks to which region the measured state belongs and correspondingly decides the next sampling instant.
\end{abstract}
\section{Introduction}
Nowadays, the use of shared networks and digital platforms for control purposes is becoming more and more ubiquitous. This has shifted the control community's research focus from periodic to aperiodic sampling techniques, which promise to reduce resource utilization (e.g. bandwidth, processing power, etc.). Arguably, event-based control is the aperiodic scheme which has attracted wider attention, with its two sub-branches being Event-Triggered Control (ETC, e.g. \cite{arzen1999,astrom2002comparison,tabuada2007etc,small_gain_robust_etc, girard2015dynamicetc}) and Self-Triggered Control (STC, e.g. \cite{velasco2003self,tosample, mazo2010stc_iss, wang2010self, small_gain_robust_etc, italy_digital_stc, tiberi_stc_perturbed, med_stc, delimpaltadakis_tac}). For an introduction to the topic, the reader is referred to \cite{2012introtoetc_stc}.

ETC and STC are sample-and-hold implementations of digital control. In ETC, intelligent sensors monitor continuously the system's state, and transmit data only when a state-dependent \textit{triggering condition} is satisfied. On the other hand, to tackle the necessity of dedicated intelligent hardware, STC has been proposed, in which the controller at each sampling time instant decides the next one based solely on present measurements. The most common way to decide the next sampling time in STC is the emulation approach: predicting conservatively when the triggering condition of a corresponding ETC scheme would be satisfied. In this way, STC provides the same performance guarantees as the underlying ETC scheme, although it generally leads to faster sampling.

Unfortunately, published work regarding STC for perturbed uncertain nonlinear systems remains very scarce. In \cite{small_gain_robust_etc}, ETC and STC schemes are designed for input-to-state stable (ISS) systems subject to disturbances, by employing a small-gain approach. To address model uncertainties, the authors consider nonlinear systems in strict-feedback form and propose a control-design procedure that compensates for the uncertain dynamics, in a way such that the previously derived STC scheme would still guarantee stability. In \cite{italy_digital_stc}, a self-triggered sampler is derived, that guarantees that the system remains in a safe set, by employing Taylor approximations of the Lyapunov function's derivative. Finally, \cite{tiberi_stc_perturbed} designs ETC and STC that guarantee uniform ultimate boundedness for perturbed uncertain systems, while \cite{med_stc} employs the small-gain approach to design STC that guarantees $\mathcal{L}_p$-stability. Alternative approaches relying on a stochastic framework and learning techniques have also been proposed, see e.g. \cite{hashimoto2020learning, probabilistic}. In contrast to the robust approaches listed earlier, they can cope with potentially unbounded disturbances, but they relax the performance guarantees to probabilistic assurances.

Here, we extend the region-based STC framework of \cite{delimpaltadakis_tac} and propose an STC scheme for general nonlinear systems with bounded disturbances/uncertainties, providing deterministic guarantees. This framework is able to emulate a wide range of triggering conditions and corresponding ETC schemes in a unified generic way. Hence, compared to the deterministic approaches listed above, which focus on emulating one class of triggering conditions and provide one specific performance specification, it is more versatile, as it can provide different robust performance guarantees (stability, safety, boundedness, etc.), depending on the ETC scheme that is emulated.

Particularly, in \cite{delimpaltadakis_tac} a region-based STC scheme for smooth nonlinear systems has been proposed, which provides inter-sampling times that lower bound the ideal inter-sampling times of an a-priori given ETC scheme. The state-space is partitioned into regions $\reg_i$, each of which is associated to a uniform inter-sampling time $\tau_i$. The regions $\reg_i$ are sets delimited by inner-approximations of \textit{isochronous manifolds} (sets composed of points in the state-space that correspond to the same ETC inter-sampling time). In a real-time implementation, the controller checks to which of the regions the measured state belongs, and correspondingly decides the next sampling time.

Here, we extend the above framework to systems with disturbances and uncertainties, which greatly facilitates the applicability of region-based STC in practice. To deal with disturbances and uncertainties in a unified way, we abstract perturbed uncertain systems by differential inclusions (DIs). Moreover, we introduce ETC notions, such as the inter-sampling time, to the DI-framework. Within the DI-framework, by employing the notion of homogeneous DIs (see \cite{bernuau2013homogeneous_diff_incl}), we are able to extend well-known significant results on ETC/STC, namely the scaling law of inter-sampling times \cite{tosample} and the homogenization procedure \cite{anta2012isochrony}, to perturbed uncertain systems. Based on these renewed results, we construct approximations of isochronous manifolds of perturbed uncertain ETC systems, thus extending region-based STC to perturbed uncertain systems. We showcase our theoretical results via simulations and comparisons with other deterministic approaches, which indicate that the proposed STC scheme shows competitive performance, while simultaneously achieving greater generality.

Apart from the above, let us emphasize that the merits of obtaining approximations of isochronous manifolds extend beyond the context of STC design, as it enables discovering relations between regions in the state-space of an ETC system and inter-sampling times. In fact, such approximations have already been used to construct advanced timing models, that capture the sampling behaviour of unperturbed homogeneous ETC systems, and are then used for traffic scheduling in networks of ETC loops (see \cite{delimpa2020traffic}). More generally, as noted in \cite{delimpaltadakis_tac}, isochronous manifolds are an inherent characteristic of any system with an output. Thus, implications of the theoretical contribution of approximating them, especially under the effect of disturbances and uncertainties, might even exceed the mere context of ETC/STC.

To summarize our contributions, in this work we:
\begin{itemize}
	\item construct a framework based on DIs that allows reasoning about perturbed uncertain ETC systems,
	\item extend important results on ETC/STC to perturbed uncertain systems, by employing the DI-framework,
	\item obtain approximations of isochronous manifolds of perturbed uncertain ETC systems,
	\item and design a robust STC scheme for perturbed uncertain nonlinear systems, which simultaneously achieves greater versatility and competitive performance, compared to the existing literature.
\end{itemize}


\section{Notation and Preliminaries}
\subsection{Notation}
We denote points in $\mathbb{R}^n$ as $x$ and their Euclidean norm as $|x|$. For vectors, we also use the notation $(x_1,x_2)=\begin{bmatrix}x_1^\top &x_2^\top\end{bmatrix}^\top$. Consider a set $\init\subseteq\real^n$. Then, $\overline{\init}$ denotes its closure, $\mathrm{int}(\init)$ its interior and $\mathrm{conv}(\init)$ its convex hull. Moreover, for any $\lambda\in\real$ we denote: $\lambda\init = \{\lambda x\in\real^n:\text{ }x\in\init\}$. 

Consider a system of ordinary differential equations (ODE):
\begin{equation} \label{ode}
\dot{\zeta}(t) = f(\zeta(t)),
\end{equation}
where $\zeta:\real\to\real^n$. We denote by $\zeta(t; t_0, \zeta_0)$ the solution of \eqref{ode} with initial condition $\zeta_0$ and initial time $t_0$. When $t_0$ (or $\zeta_0$) is clear from the context, then it is omitted, i.e. we write $\zeta(t;\zeta_0)$ (or $\zeta(t)$).

Consider the differential inclusion (DI):
\begin{equation} \label{di}
\dot{\zeta}(t) \in F(\zeta(t)),
\end{equation}
where $\zeta:\real\to\real^n$ and $F:\real^n\rightrightarrows\real^n$ is a set-valued map. In contrast to ODEs, which under mild assumptions obtain unique solutions given an initial condition, DIs generally obtain multiple solutions for each initial condition, which might even not be defined for all time. We denote by $\zeta(t;\zeta_0)$ any solution of \eqref{di} with initial condition $\zeta_0$. Moreover, $\mathcal{S}_F([0,T];\init)$ denotes the set of all solutions of \eqref{di} with initial conditions in $\init\subseteq\real^n$, which are defined on $[0,T]$. Thus, the \emph{reachable set} from $\init\subseteq\real^n$ of \eqref{di} at time $T\geq0$ is defined as:
\begin{equation*}
\reach^F_T(\init) = \{\xi(T;\xi_0):\text{ }\text{ }\xi(\cdot;\xi_0)\in\mathcal{S}_F([0,T];\init)\}.
\end{equation*} 
Likewise, the \emph{reachable flowpipe} from $\init\subseteq\real^n$ of \eqref{di} in the interval $[\tau_1,\tau_2]$ is $\reach^F_{[\tau_1,\tau_2]}(\init)=\bigcup\limits_{t\in[\tau_1,\tau_2]}\reach^F_t(\init)$.

\subsection{Homogeneous Systems and Differential Inclusions}
Here, we focus on the classical notion of homogeneity, with respect to the standard dilation. For the general definition and more information the reader is referred to \cite{bernuau2013homogeneous_diff_incl} and \cite{kawski_homogeneity}.
\begin{definition}[Homogeneous functions and set-valued maps]
	Consider a function $f:\real^n\to\real^m$ (or a set-valued map $F:\real^n\rightrightarrows\real^m$). We say that $f$ (or $F$) is homogeneous of degree $\alpha\in\real$, if for all $x\in\real^n$ and any $\lambda>0$: $f(\lambda x) = \lambda^{\alpha+1}f(x)$ (respectively $F(\lambda x) = \lambda^{\alpha+1}F(x)$).
\end{definition}
Correspondingly, a system of ODEs \eqref{ode} or a DI \eqref{di} is called homogeneous of degree $\alpha\in\real$ if the corresponding vector field or set-valued map is homogeneous of the same degree. For homogeneous ODEs or DIs, the following scaling property of solutions holds:
\begin{proposition}[Scaling Property \cite{kawski_homogeneity, bernuau2013homogeneous_diff_incl}]
	Let the system of ODEs \eqref{ode} be homogeneous of degree $\alpha\in\real$. Then, for any $\zeta_0\in\real^n$ and any $\lambda>0$:
	\begin{equation}\label{scaling_ode}
		\zeta(t;\lambda\zeta_0) = \lambda\zeta(\lambda^{\alpha}t;\zeta_0).
	\end{equation} 
	Now, let DI \eqref{di} be homogeneous of degree $\alpha\in\real$. Then, for any $\init\subseteq\real^n$ and any $\lambda>0$:	
	\begin{equation}\label{scaling_di}
		\reach^F_t(\lambda\init) = \lambda\reach^F_{\lambda^{\alpha}t}(\init).
	\end{equation}
\end{proposition}

\subsection{Event-Triggered Control Systems}
Consider the control system with state-feedback:
\begin{equation}\label{ct_sys}
	\dot{\zeta}(t) = f\Big(\zeta(t),\upsilon(\zeta(t))\Big),
\end{equation}
where $\zeta:\real\to\real^n$, $f:\real^n\times\real^{m_u}\to\real^n$, and $\upsilon:\real^n\to\real^{m_u}$ is the control input. In any sample-and-hold scheme, the control input is updated on sampling time instants $t_i$ and held constant between consecutive sampling times:
\begin{equation*}
	\dot{\zeta}(t) = f\Big(\zeta(t),\upsilon(\zeta(t_i))\Big), \quad t\in[t_i,t_{i+1}).
\end{equation*}
If we define the \textit{measurement error} as the difference between the last measurement and the present state: 
\begin{equation*}
	\varepsilon_{\zeta}(t):=\zeta(t_i)-\zeta(t), \quad t\in[t_i,t_{i+1}),
\end{equation*}
then the sample-and-hold system can be written as:
\begin{equation}\label{snh_sys}
\dot{\zeta}(t) = f\Big(\zeta(t),\upsilon(\zeta(t)+\varepsilon_{\zeta}(t))\Big), \quad t\in[t_i,t_{i+1}).
\end{equation}
Notice that the error $\varepsilon_{\zeta}(t)$ resets to zero at each sampling time. In ETC, the sampling times are determined by:
\begin{equation}\label{trig_cond}
	t_{i+1} = t_i + \inf\{t>0:\text{ }\phi(\zeta(t;x_i),\varepsilon_{\zeta}(t))\geq 0\},
\end{equation}
and $t_0=0$, where $x_i\in\real^n$ is the previously sampled state, $\phi(\cdot,\cdot)$ is the \textit{triggering function}, \eqref{trig_cond} is the \textit{triggering condition} and $t_{i+1}-t_i$ is called \textit{inter-sampling time}. Each point $x\in\real^n$ corresponds to a specific inter-sampling time, defined as:
\begin{equation}\label{intersampling_time}
	\tau(x):=\inf\{t>0:\text{ }\phi(\zeta(t;x),\varepsilon_{\zeta}(t))\geq 0\}.
\end{equation}

Finally, since $\dot{\varepsilon}_{\zeta}(t)=-\dot{\zeta}(t)$, we can write the dynamics of the extended ETC closed loop in a compact form:
\begin{equation} \label{etc_system}
\begin{aligned}
&\dot{\xi}(t)= \begin{bmatrix} f\Big(\zeta(t),\upsilon(\zeta(t)+\varepsilon_{\zeta}(t))\Big)\\
-f\Big(\zeta(t),\upsilon(\zeta(t)+\varepsilon_{\zeta}(t))\Big) \end{bmatrix}=f_e(\xi(t)), \text{ } t \in [t_i, t_{i+1})\\
&\xi(t_{i+1}^+)=\begin{bmatrix}
\zeta(t_{i+1}^-)\\ 0
\end{bmatrix},
\end{aligned}
\end{equation}
where $\xi = (\zeta,\varepsilon_{\zeta})\in\real^{2n}$. At each sampling time $t_i$, the state of \eqref{etc_system} becomes $\xi_i=(x_i,0)$. Thus, since we are interested in intervals between consecutive sampling times, instead of writing $\phi\Big(\xi(t;(x_i,0))\Big)$ (or $\tau\Big((x_i,0)\Big)$), we abusively write $\phi(\xi(t;x_i))$ (or $\tau(x_i)$) for convenience. Between two consecutive sampling times, the triggering function starts from a negative value $\phi(\xi(t_i;x_i))<0$, and stays negative until $t_{i+1}^-$, when it becomes zero.  Triggering functions are designed such that the inequality $\phi(\xi(t;x_i))\leq0$ implies certain performance guarantees (e.g. stability). Thus, sampling times are defined in a way (see \eqref{trig_cond}) such that $\phi(\xi(t))\leq0$ for all $t\geq0$, which implies that the performance specifications are met at all time.

\subsection{Self-Triggered Control: Emulation Approach}
The emulation approach to STC entails providing conservative estimates of a corresponding ETC scheme's inter-sampling times, based solely on the present measurement $x_i$: 
\begin{equation}\label{stc bounds etc}
	\tau^{\downarrow}(x_i)\leq\tau(x_i),
\end{equation} 
where $\tau^{\downarrow}(\cdot)$ denotes STC inter-sampling times. This guarantees that the triggering function of the emulated ETC remains negative at all time, i.e. STC provides the same guarantees as the emulated ETC. Thus, STC inter-sampling times should be no larger than ETC ones, but as large as possible in order to reduce resource utilization. Finally, infinitely fast sampling (Zeno phenomenon) should be avoided, i.e. $\inf\limits_{x}\tau^{\downarrow}(x)\geq\epsilon>0$.

\section{Problem Statement}
In \cite{delimpaltadakis_tac}, for a system \eqref{etc_system}, given a triggering function $\phi(\cdot)$ and a finite set of arbitrary user-defined times $\{\tau_1,\tau_2,\dots,\tau_q\}$ (where $\tau_i<\tau_{i+1}$), which serve as STC inter-sampling times, the state-space of the original system \eqref{snh_sys} is partitioned into regions $\reg_i\subset\real^n$ such that:
\begin{equation} \label{regions-times}
\forall x \in \mathcal{R}_i:\quad \tau_i\leq\tau(x),
\end{equation}
where $\tau(x)$ denotes ETC inter-sampling times corresponding to the given triggering function $\phi(\cdot)$. The region-based STC protocol operates as follows:
\begin{enumerate}
	\item Measure the current state $\xi(t_k)=(x_k,0)$.
	\item Check to which of the regions $\mathcal{R}_i$ the point $x_k$ belongs.
	\item If $x_k\in\mathcal{R}_i$, set the next sampling time to $t_{k+1} = t_k + \tau_i$.
\end{enumerate}

As mentioned in the introduction, we aim at extending the STC technique of \cite{delimpaltadakis_tac} to systems with disturbances and uncertainties. Thus, we consider perturbed/uncertain ETC systems, written in the compact form:
\begin{equation}\label{perturbed_etc_sys}
	\dot{\xi}(t)= \begin{bmatrix} f\Big(\zeta(t),\upsilon(\zeta(t)+\varepsilon_{\zeta}(t)), d(t)\Big)\\
	-f\Big(\zeta(t),\upsilon(\zeta(t)+\varepsilon_{\zeta}(t)), d(t)\Big) \end{bmatrix}=f_e(\xi(t),d(t)),
\end{equation}
where $d:\real\to\real^{m_d}$ is an unknown signal (e.g. disturbance, model uncertainty, etc.), and assume that a triggering function $\phi(\xi(t))$ is given. 
\begin{assumption}\label{assum1}
	For the remainder of the article we assume the following:
	\begin{enumerate}
		\item The function $f_e(\cdot,\cdot)$ is locally bounded and continuous with respect to all of its arguments. \label{bounded_fe}
		\item For all $t\geq 0$: $d(t)\in\Delta$, where $\Delta\subset\real^{m_d}$ is convex, compact and non-empty. \label{boundedness_assumption}
		\item The function $\phi(\cdot)$ is continuously differentiable.\label{trig_fun_differentiability}
		\item For all $\xi_0=(x_0,0)\in\real^{2n}\mathrm{:}$ $\phi(\xi_0)<0$. Moreover, for any compact set $K\subset\real^n$ there exists $\epsilon_K>0$ such that for all $x_0\in K$ and any $d(t)$, with $d(t)\in\Delta$ for all $t\geq0$, $\phi(\xi(t;\xi_0))<0$ for all $t\in[0,\epsilon_K)$. \label{trig_fun_negativity}
	\end{enumerate}
\end{assumption}
The problem statement of this work is:
\begin{problem statement}
	Given a system \eqref{perturbed_etc_sys} and a triggering function $\phi(\cdot)$, which satisfy Assumption \ref{assum1}, and a predefined finite set of times $\{\tau_1\dots,\tau_q\}$ (with $\tau_i<\tau_{i+1}$), derive regions $\reg_i\subset\real^n$ that satisfy \eqref{regions-times}.
\end{problem statement}
\begin{remark}
	In region-based STC, the Zeno phenomenon is ruled out by construction, since region-based STC inter-sampling times are lower-bounded: $\tau^\downarrow(x)\geq\min\limits_{i}\tau_i=\tau_1$.
\end{remark}
Items \ref{bounded_fe} and \ref{boundedness_assumption} of Assumption \ref{assum1} impose the satisfaction of the standard assumptions of differential inclusions on the DIs that we construct later (see \eqref{diff_incl_1}). These assumptions ensure existence of solutions for all initial conditions (see \cite{bernuau2013homogeneous_diff_incl} and \cite{filippov2013differential} for more details). Note that assuming convexity of $\Delta$ is not restrictive, since in the case of a non-convex $\Delta$ we can consider the closure of its convex hull and write $d(t)\in\overline{\mathrm{conv}(\Delta)}$ for all $t\geq0$. Finally, item \ref{trig_fun_differentiability} is employed in the proof of Lemma \ref{bounding_lemma}, while item \ref{trig_fun_negativity} ensures that the emulated ETC associated to the given triggering function does not exhibit Zeno behaviour for any given bounded state.
\begin{remark}\label{trig_functions_remark}
	The triggering function should be chosen to be robust to disturbances/uncertainties, such that the emulated ETC does not exhibit Zeno behaviour. 
	Examples of such robust triggering functions are: 
	\begin{itemize}
		\item Lebesgue sampling (e.g. \cite{arzen1999,tiberi_stc_perturbed}): $\phi(\xi(t)) = |\varepsilon_{\zeta}(t)|^2-\epsilon^2$, where $\epsilon>0$. 
		\item Mixed-Triggering (e.g. \cite{small_gain_robust_etc}): $\phi(\xi(t)) = |\varepsilon_{\zeta}(t)|^2-\sigma|\zeta(t)|^2-\epsilon^2$, where $\sigma>0$ is appropriately chosen and $\epsilon>0$.
	\end{itemize}
	Both functions satisfy Assumption \ref{assum1}.
\end{remark}

\section{Overview of \cite{delimpaltadakis_tac}} \label{previous_work_section}
In this section, we give a brief overview of \cite{delimpaltadakis_tac}. First, we focus on homogeneous ETC systems and how, exploiting approximations of their isochronous manifolds, a state-space partitioning into regions $\reg_i$ can be derived. Afterwards, we recall how these results can be generalized to general nonlinear systems, by employing a homogenization procedure.

\subsection{Homogeneous ETC Systems, Isochronous Manifolds and State-Space Partitioning}
\begin{figure*}[h!]
	\centering
	\begin{subfigure}[t]{0.31\textwidth}
		\centering
		\includegraphics[width=0.7\textwidth]{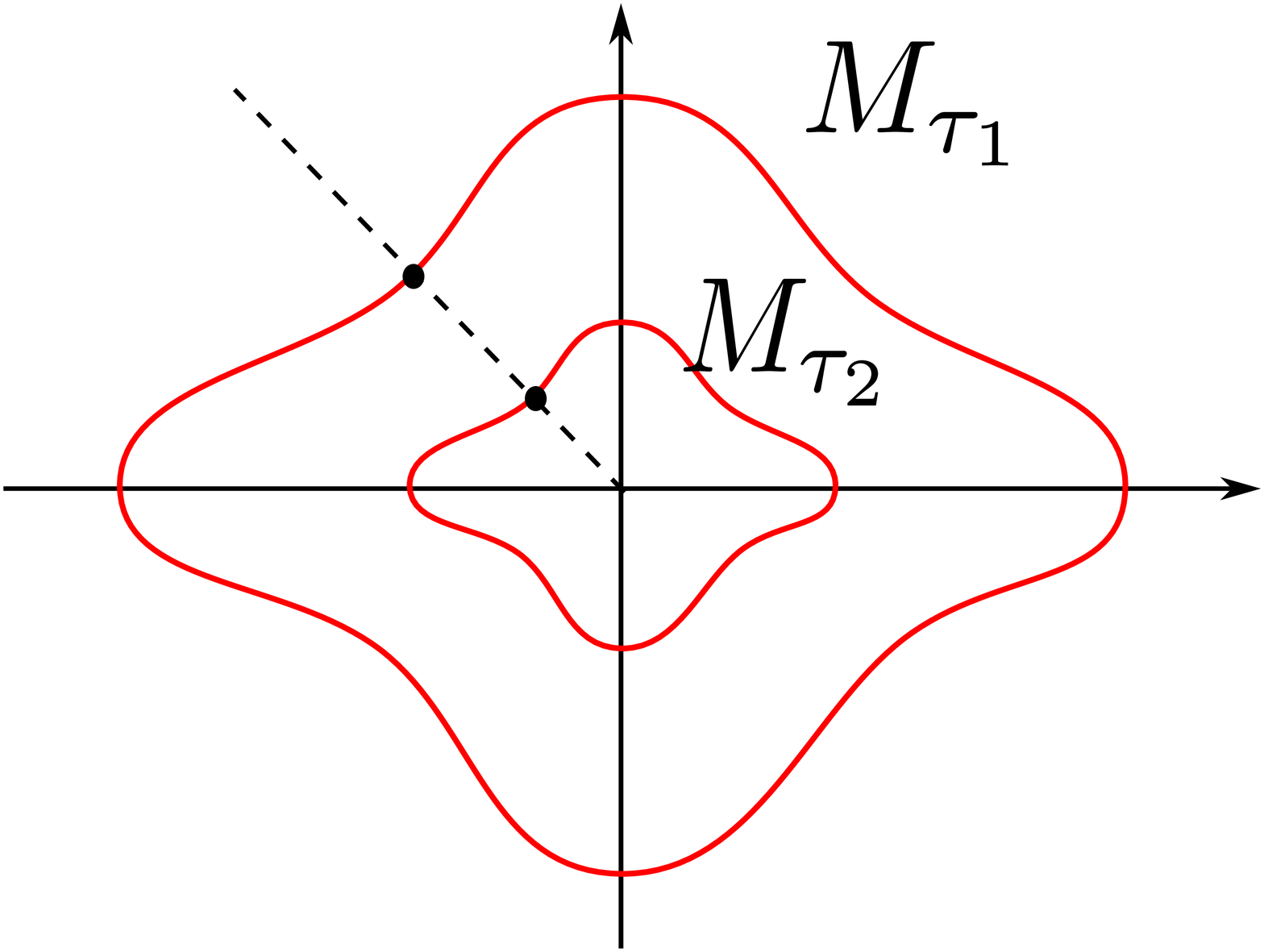}
		\caption{Isochronous manifolds of a homogeneous ETC system for times $\tau_1<\tau_2$.}
		\label{two_manifolds_fig}
	\end{subfigure}\quad
	\begin{subfigure}[t]{0.31\textwidth}
		\centering
		\includegraphics[width=0.8\textwidth]{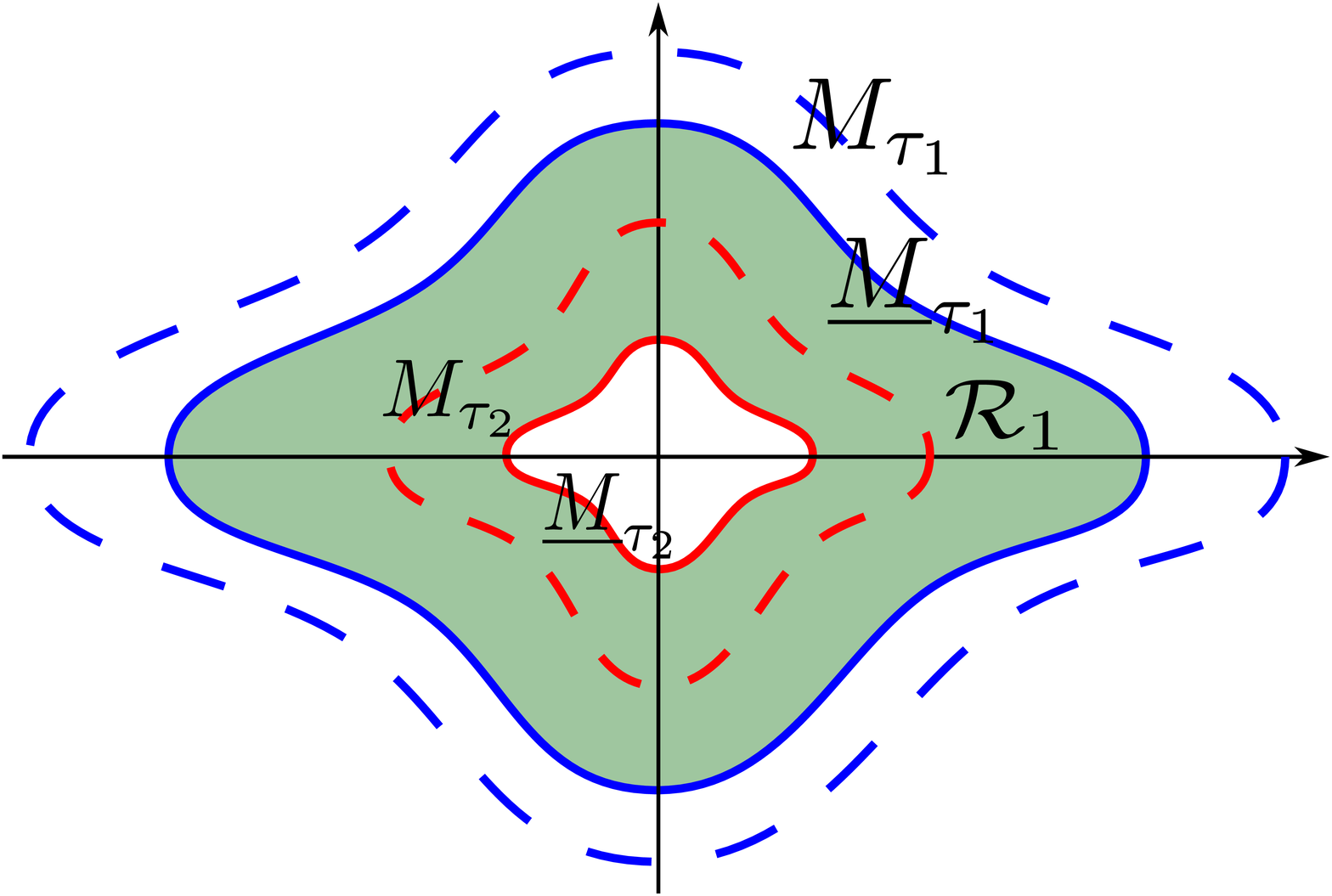}
		\caption{Isochronous manifolds $M_{\tau_i}$ (dashed lines), and inner-approximations $\underline{M}_{\tau_i}$ (solid lines). The filled region represents $\reg_1$.}
		\label{discr_approx}
	\end{subfigure}\quad
	\begin{subfigure}[t]{0.31\textwidth}
		\centering
		\includegraphics[width=0.7\textwidth]{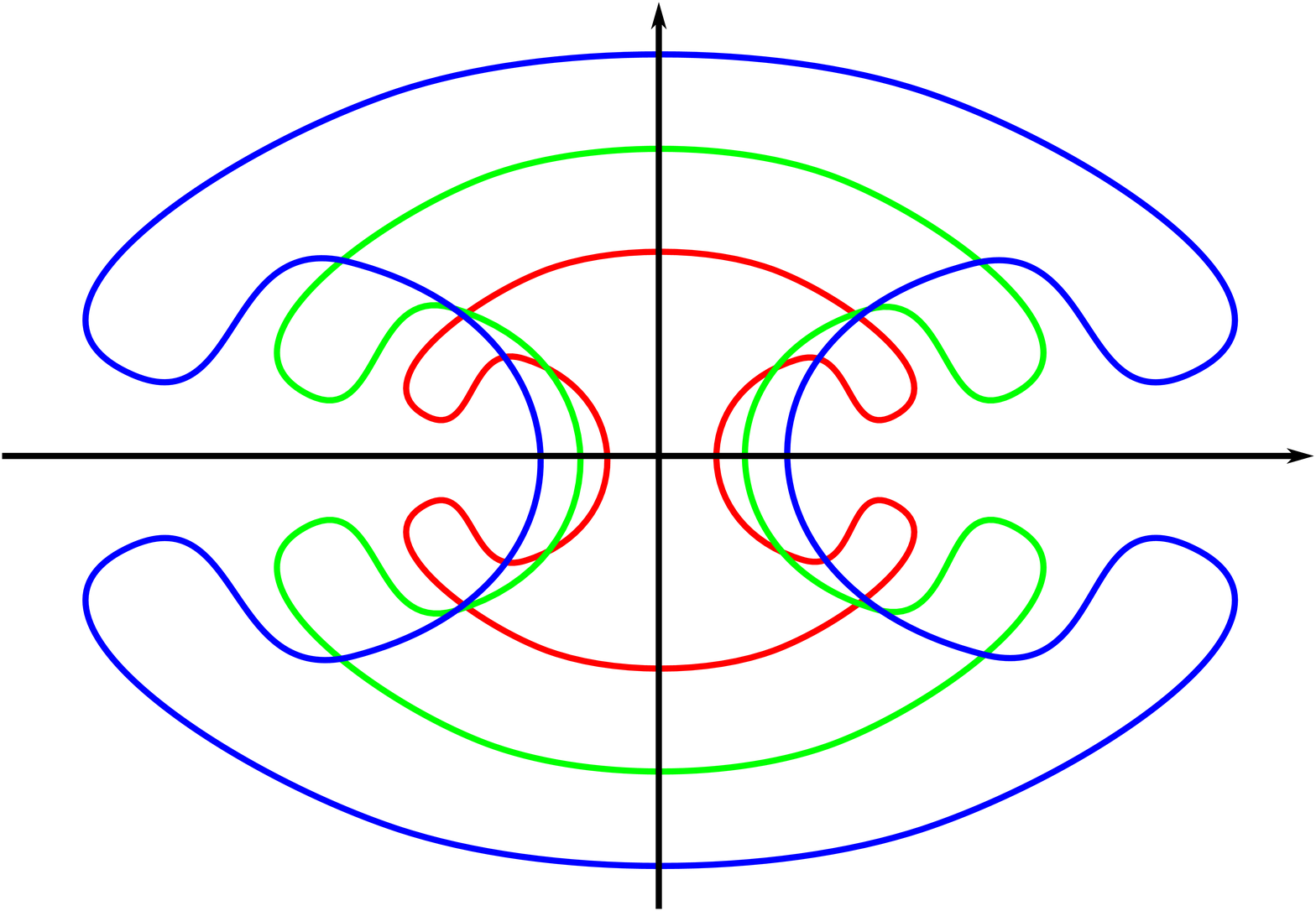}
		\caption{If inner-approximations of isochronous manifolds did not satisfy properties \ref{prop3}-\ref{prop4} from Proposition \ref{manifold properties}, then the regions $\reg_i$ could intersect with each other.}
		\label{bad_discr}
	\end{subfigure}
	\caption{Isochronous manifolds and inner-approximations.}
\end{figure*}
An important property of homogeneous ETC systems is the scaling of inter-sampling times:
\begin{theorem}[Scaling of ETC Inter-Sampling Times \cite{tosample}]
	Consider an ETC system \eqref{etc_system} and a triggering function, homogeneous of degree $\alpha$ and $\theta$ respectively. Then, for all $x\in\real^n$ and $\lambda>0$:
	\begin{equation}\label{intersampling_scaling}
	\tau(\lambda x) = \lambda^{-\alpha}\tau(x),
	\end{equation}
	where $\tau(\cdot)$ is defined in \eqref{intersampling_time}.
\end{theorem}
Thus, for homogeneous ETC systems with degree $\alpha>0$, along a ray that starts from the origin (\textit{homogeneous ray}), inter-sampling times become larger for points closer to the origin. Scaling law \eqref{intersampling_scaling} is a direct consequence of the system's and triggering function's homogeneity, since \eqref{scaling_ode} implies that:
\begin{equation}\label{trig_function_scaling}
	\phi(\xi(t;\lambda x))=\phi(\lambda\xi(\lambda^\alpha t;x))=\lambda^{\theta+1}\phi(\xi(\lambda^\alpha t;x)).
\end{equation}

In \cite{delimpaltadakis_tac}, the scaling law \eqref{intersampling_scaling} is combined with inner-approximations of \textit{isochronous manifolds}, a notion firstly introduced in \cite{anta2012isochrony}. Isochronous manifolds are sets of points with the same inter-sampling time:
\begin{definition}[Isochronous Manifolds]
	Consider an ETC system \eqref{etc_system}. The set $M_{\tau_{\star}}=\{x\in\mathbb{R}^n : \tau(x)=\tau_{\star}\}$, where $\tau(x)$ is defined in \eqref{intersampling_time}, is called isochronous manifold of time $\tau_{\star}$.
	\label{manifold definition}
\end{definition}
For homogeneous systems, the scaling law \eqref{intersampling_scaling} implies that isochronous manifolds satisfy the following properties:
\begin{proposition}[\hspace{1sp}\cite{delimpaltadakis_tac,anta2012isochrony}]\label{manifold properties}
	Consider an ETC system \eqref{etc_system} and a triggering function, homogeneous of degree $\alpha>0$ and $\theta>0$ respectively, and let Assumption \ref{assum1} hold. Then:
	\begin{enumerate}
		\item For any time $\tau_{\star}>0$, there exists an isochronous manifold $M_{\tau_{\star}}$. \label{prop1}
		\item Isochronous manifolds are hypersurfaces of dimension $n-1$.\label{prop2}
		\item Each homogeneous ray intersects an isochronous manifold $M_{\tau_{\star}}$ only at one point.\label{prop3}
		\item \label{prop4}Given two isochronous manifolds $M_{\tau_1}$, $M_{\tau_2}$ with $\tau_1<\tau_2$, on every homogeneous ray $M_{\tau_1}$ is further away from the origin compared to $M_{\tau_2}$, i.e. for all $x\in M_{\tau_1}$:
		\begin{equation*}
		\begin{aligned}
			&\bullet\exists! \lambda_x\in(0,1) \text{ s.t. } \lambda_x x \in M_{\tau_{2}},\\
			&\bullet\not\exists \kappa_x\geq 1 \text{ s.t. } \kappa_x x \in M_{\tau_{2}}.
		\end{aligned}	
		\end{equation*}
	\end{enumerate}
\end{proposition}
Properties \ref{prop2}-\ref{prop4} from Proposition \ref{manifold properties} are illustrated in Fig. \ref{two_manifolds_fig}.
Now, consider the region $R_1$ between isochronous manifolds $M_{\tau_1}$ and $M_{\tau_2}$ in Fig. \ref{two_manifolds_fig}. The scaling law $\eqref{intersampling_scaling}$ directly implies that for all $x\in R_1$: $\tau_1\leq\tau(x)$, i.e. \eqref{regions-times} is satisfied.
Thus, if isochronous manifolds could be computed, then the state-space could be partitioned into the regions delimited by isochronous manifolds and the region-based STC scheme would be enabled.
\subsection{Inner-Approximations of Isochronous Manifolds}\label{inner_approx_section}
Since isochronous manifolds cannot be computed analytically, in \cite{delimpaltadakis_tac} inner-approximations $\underline{M}_{\tau_i}$ of isochronous manifolds $M_{\tau_i}$ are derived in an analytic form (see Fig. \ref{discr_approx}). Again due to the scaling law, for the region $\reg_1$ between two inner-approximations $\underline{M}_{\tau_1}$ and $\underline{M}_{\tau_2}$ (with $\tau_1<\tau_2$) it holds that $\tau_1\leq\tau(x)$ for all $x\in\reg_1$. Hence, given a set of times $\{\tau_1,\dots,\tau_q\}$, the state-space is partitioned into regions $\reg_i$ delimited by these inner-approximations. As noted in \cite{delimpaltadakis_tac}, it is crucial that approximations $\underline{M}_{\tau_i}$ have to satisfy the same properties as isochronous manifolds, mentioned in Proposition \ref{manifold properties}. For example, 
if approximations $\underline{M}_{\tau_i}$ did not satisfy properties \ref{prop3}-\ref{prop4}, then $\reg_i$ could potentially intersect with each other and be ill-defined (see Fig. \ref{bad_discr}).

To derive the inner-approximations, the triggering function $\phi(\xi(t;x))$ is upper-bounded by a function $\mu(x,t)$ with linear dynamics, that satisfies certain conditions. Then, the sets $\underline{M}_{\tau_i}=\{x\in\real^n:\text{ }\mu(x,\tau_i)=0\}$ are proven to be inner-approximations of isochronous manifolds $M_{\tau_i}$. The sufficient conditions that $\mu(x,t)$ has to satisfy in order for its zero-level sets to be inner-approximations of isochronous manifolds and satisfy the properties mentioned in Proposition \ref{manifold properties} are summarized in the following theorem:
\begin{theorem}[\hspace{1sp}\cite{delimpaltadakis_tac}] \label{theorem1_tac}
	Consider an ETC system \eqref{etc_system} and a triggering function $\phi(\cdot)$, homogeneous of degrees $\alpha>0$ and $\theta>0$, respectively. Let $\mu:\mathbb{R}^n\times\mathbb{R}^+\rightarrow\mathbb{R}$ be a function that satisfies:
	\begin{subequations}
		\begin{align}
		&\mu(x,0)<0, \quad \forall x \in \mathbb{R}^n\setminus\{0\}, \label{bound init cond}\\
		&\mu(x,t) \geq \phi(\xi(t;x)), \quad \forall t \in [0,\tau(x)] \text{ }\mathrm{and} \text{ }\forall x \in \mathbb{R}^n\setminus\{0\}, \label{bound time validity}\\
		&\mu(\lambda x,t) = \lambda^{\theta+1}\mu(x,\lambda^{\alpha}t), \quad \forall t,\lambda>0 \text{ }\mathrm{and}\text{ } \forall x \in \mathbb{R}^n\setminus\{0\}, \label{scaling of bound}\\
		&\forall x \in \mathbb{R}^n\setminus\{0\}: \quad \exists!\tau_x>0  \text{ such that } \mu(x,\tau_x)=0. \label{one-zero-crossing of bound}
		\end{align}
		\label{bound requirements}
	\end{subequations}
	The sets $\underline{M}_{\tau_{\star}} := \{x\in\mathbb{R}^n: \mu(x,\tau_{\star})=0\}$ are inner-approximations of isochronous manifolds $M_{\tau_{\star}}$. Moreover, the sets $\underline{M}_{\tau_{\star}}$ satisfy the properties mentioned in Proposition \ref{manifold properties}.
\end{theorem}
Let us briefly explain what is the intuition behind this theorem. Since \eqref{bound init cond} and \eqref{bound time validity} hold, if we denote by $\tau^{\downarrow}(x):=\inf\{t>0:\mu(x,t)=0\}$ (from \eqref{one-zero-crossing of bound} we know that it exists), then $\tau^{\downarrow}(x)\leq\tau(x)$. Note that it is important that inequality \eqref{bound time validity} extends at least until $t=\tau(x)$, in order for $\tau^{\downarrow}(x)\leq\tau(x)$. Then, by the scaling law \eqref{intersampling_scaling}, we have that the set $\underline{M}_{\tau_{\star}}=\{x\in\real^n:\tau^{\downarrow}(x)=\tau_{\star}\}$ is an inner-approximation of the isochronous manifold $M_{\tau_\star}$. Moreover, since for each $x$ the equation $\mu(x,t)=0$ has a unique solution w.r.t. $t$ (from \eqref{one-zero-crossing of bound}), we get that $\underline{M}_{\tau_{\star}}\equiv\{x\in\mathbb{R}^n: \mu(x,\tau_{\star})=0\}$. Finally, condition \eqref{scaling of bound} implies that $\tau^{\downarrow}(\lambda x)=\lambda^{-\alpha}\tau^{\downarrow}(x)$ (observe the similarity between \eqref{scaling of bound} and \eqref{trig_function_scaling}), which in turn implies that the sets $\underline{M}_{\tau_{\star}}$ satisfy the properties of Proposition \eqref{manifold properties}. We do not elaborate more on the technical details here (e.g. how is the bounding carried out), since we address these later in the document, where we extend the theoretical results of \cite{delimpaltadakis_tac} to perturbed/uncertain systems.

\subsection{Homogenization of Nonlinear Systems and Region-Based STC}

To exploit aforementioned properties of homogeneous systems, the homogenization procedure proposed by \cite{anta2012isochrony} is employed. Any non-homogeneous system \eqref{etc_system} is rendered homogeneous of degree $\alpha>0$, by embedding it into $\real^{2n+1}$ and adding a dummy variable $w$:
\begin{equation}\label{homogenized_etc_system}
\begin{bmatrix}
\dot{\xi}\\\dot{w}
\end{bmatrix} = \begin{bmatrix}
w^{\alpha+1}f_e(w^{-1}\xi)\\0
\end{bmatrix} = \tilde{f}_e(\xi,w)
\end{equation}
The same can be done for non-homogeneous triggering functions $\tilde{\phi}(\xi,w)=w^{\theta+1}\phi(w^{-1}\xi)$. Notice that the trajectories of the original ETC system \eqref{etc_system} with initial condition $(x_0,e_0)\in\real^{2n}$ coincide with the trajectories of the homogenized one \eqref{homogenized_etc_system} with initial condition $(x_0,e_0,1)\in\real^{2n+1}$, projected to the $\xi$-variables. The same holds for a homogenized triggering function. Thus, the inter-sampling times $\tau(x_0)$ of system \eqref{etc_system} with triggering function $\phi(\cdot)$ coincide with the inter-sampling times $\tau\Big((x_0,1)\Big)$ of \eqref{homogenized_etc_system} with triggering function $\tilde{\phi}(\cdot)$.

Consequently, if the original system (or the triggering function) is non-homogeneous, then first it is rendered homogeneous via the homogenization procedure \eqref{homogenized_etc_system}. Afterwards, inner-approximations of isochronous manifolds for the homogenized system \eqref{homogenized_etc_system} are derived. Since trajectories of the original system are mapped to trajectories on the $w=1$-plane of the homogenized one (i.e. the state-space of the original system is mapped to the $w=1$-plane), to determine the inter-sampling time $\tau_i$ of a state $x_0\in\real^n$, one has to check to which region $\reg_i\subset\real^{n+1}$ the point $(x_0,1)$ belongs. For an illustration, see Figure \ref{w_1_discr}: 
e.g. given a state $x_0\in\real^n$, if $(x_0,1)\in\real^{n+1}$ lies on the cyan segment (i.e. it is contained in $\reg_1$), then the STC inter-sampling time that is assigned to $x_0$ is $\tau^\downarrow(x_0)=\tau_1$.

Note that, here, it suffices to inner-approximate the isochronous manifolds of \eqref{homogenized_etc_system} only in the subspace $w>0$, since we only care about determining regions $\reg_i$ for points $(x_0,1)\in\real^{2n}$. Thus, the conditions of Theorem \ref{theorem1_tac} can be relaxed so that they hold only in the subspace $w>0$, i.e. for all $(x,w)\in(\real^n\setminus\{0\})\times\real_+$. 
\begin{figure}[!h]
	\centering
	\includegraphics[width=2in]{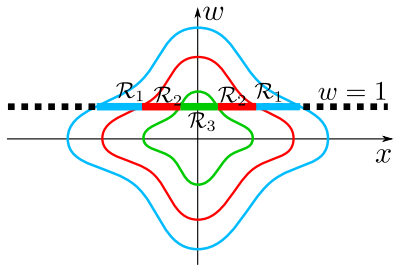}
	\caption{Inner-approximations of isochronous manifolds (coloured curves) for a homogenized system \eqref{homogenized_etc_system} and the regions $\reg_i$ between them. The coloured segments on the $w=1$-plane represent the subsets of the hyperplane $w=1$ (i.e. the subsets of the original state space) that are contained in the regions $\reg_i$ and are associated to the corresponding inter-sampling times $\tau_i$.}
	\label{w_1_discr}
\end{figure} 

\section{Perturbed/Uncertain ETC Systems as Differential Inclusions}
In this section, we show how a general perturbed/uncertain nonlinear system \eqref{perturbed_etc_sys}, satisfying Assumption \ref{assum1}, can be abstracted by a homogeneous DI. Moreover, we extend the notion of inter-sampling times in the context of DIs and show that scaling law \eqref{intersampling_scaling} holds for inter-sampling times of homogeneous DIs. These results are used afterwards in Section \ref{region_based_stc_section}, to derive inner-approximations of isochronous manifolds of perturbed/uncertain systems \eqref{perturbed_etc_sys}, and thus enable the region-based STC scheme.

\subsection{Abstractions by Differential Inclusions}
Notice that, since system \eqref{perturbed_etc_sys} is a time-varying system, many notions that we introduced before for time-invariant systems are now ill-defined. For example, depending on the realization of the unknown signal $d(t)$, a sampled state $x\in\real^n$ can correspond to different inter-sampling times, i.e. definition \eqref{intersampling_time} is ill-posed. However, employing item \ref{boundedness_assumption} of Assumption \ref{assum1} and the notion of differential inclusions, we can abstract the behaviour of the family of systems \eqref{perturbed_etc_sys} and remove such dependencies. In particular, system \eqref{perturbed_etc_sys} can be abstracted by the following differential inclusion:
\begin{equation}\label{diff_incl_1}
\dot{\xi}(t) \in F(\xi(t)) := \{f_e(\xi(t),d(t)):\text{ }d(t)\in\Delta\}.
\end{equation}
For DI \eqref{diff_incl_1} (i.e. for the family of systems \eqref{perturbed_etc_sys}), the inter-sampling time $\tau(x)$ of a point $x\in\real^n$ can now be defined as the worst-case possible inter-sampling time of $x$, under any possible signal $d(t)$ satisfying Assumption \ref{assum1}:
\begin{definition}[Inter-sampling Times of DI]
	Consider the family of systems \eqref{perturbed_etc_sys}, the DI \eqref{diff_incl_1} abstracting them, and a triggering function $\phi:\real^{2n}\to\real$. Let Assumption \ref{assum1} hold. For any point $x\in\real^n$, we define its inter-sampling time as:
	\begin{equation}\label{intersampling_times_di}
	\tau(x):=\inf\bigg\{t>0:\text{ }\sup\Big\{\phi\Big(\reach^F_t((x,0))\Big)\Big\}\geq 0\bigg\},
	\end{equation}
\end{definition}
Note that we have already emphasized that we consider initial conditions $(x,0)\in\real^{2n}$, since at any sampling time the measurement error $\varepsilon_\zeta=0$. Finally, now that inter-sampling times of systems \eqref{perturbed_etc_sys} abstracted by DIs are well-defined, we can accordingly re-define isochronous manifolds for families of such systems as: $M_{\tau_{\star}}=\{x\in\mathbb{R}^n : \tau(x)=\tau_{\star}\}$, where $\tau(x)$ is defined in \eqref{intersampling_times_di}.

\subsection{Homogenization of Differential Inclusions and Scaling of Inter-Sampling Times}\label{tac_homogenization}
As previously mentioned, the scaling law of inter-sampling times \eqref{intersampling_scaling} for homogeneous systems is of paramount importance for the approach of \cite{delimpaltadakis_tac}. We show that a similar result can be derived for inter-sampling times \eqref{intersampling_times_di} of DIs. First, observe that DI \eqref{diff_incl_1} can be rendered homogeneous of degree $\alpha>0$, by slightly adapting the homogenization procedure \eqref{homogenized_etc_system} as follows:
\begin{equation}\label{main_diff_incl}
\begin{bmatrix}
\dot{\xi}(t)\\ \dot{w}(t)
\end{bmatrix}\in \tilde{F}(\xi(t),w(t)), 
\end{equation}
where $\tilde{F}(\xi,w):=\begin{bmatrix}
\{w^{\alpha+1}f_e(w^{-1}\xi,d(t)):\text{ }d(t)\in\Delta\}\\\{0\}
\end{bmatrix}$. Indeed, $\tilde{F}(\cdot,\cdot)$ is homogeneous of degree $\alpha$.
Recall that the same can be done for a non-homogeneous triggering function: 
\begin{equation}\label{homogenized_phi}
	\tilde{\phi}(\xi,w)=w^{\theta+1}\phi(w^{-1}\xi).
\end{equation} 
Again, trajectories and flowpipes of \eqref{diff_incl_1} with initial condition $(x_0,e_0)\in\real^{2n}$ coincide with the projection to the $\xi$-variables of trajectories of \eqref{main_diff_incl} with initial condition $(x_0,e_0,1)\in\real^{2n+1}$. This implies that the inter-sampling time $\tau(x_0)$ for DI \eqref{diff_incl_1} with triggering function $\phi(\cdot)$, defined as in \eqref{intersampling_times_di}, is the same as the inter-sampling time $\tau\Big((x_0,1)\Big)$ for DI \eqref{main_diff_incl} with triggering function $\tilde{\phi}(\cdot)$.

Given the above, by employing the scaling property \eqref{scaling_di} of flowpipes of homogeneous DIs, we can prove that the scaling law holds for inter-sampling times of DIs \eqref{main_diff_incl}:
\begin{theorem}\label{di_intersampling_scaling_theorem}
	Consider DI \eqref{main_diff_incl}, the triggering function $\tilde{\phi}(\cdot)$ from \eqref{homogenized_phi}, and let Assumption \ref{assum1} hold. The inter-sampling time $\tau\Big((x,w)\Big)$, where $(x,w)\in\real^{n+1}$, scales for any $\lambda>0$ as:
	\begin{equation}
		\tau\Big(\lambda (x,w)\Big) = \lambda^{-\alpha}\tau\Big((x,w)\Big),
	\end{equation}
	where $\tau(\cdot)$ is defined in \eqref{intersampling_times_di}.
\end{theorem}
\begin{proof}
	See Appendix.
\end{proof}
For an example of how DIs and triggering functions are homogenized, the reader is referred to Section \ref{numerical example section}.

\section{Region-Based STC for Perturbed/Uncertain Systems} \label{region_based_stc_section}
In this section, we use the previous derivations about differential inclusions to inner-approximate isochronous manifolds of perturbed/uncertain systems, by adapting the technique of \cite{delimpaltadakis_tac}. Using the derived inner-approximations, the state-space partitioning into regions $\reg_i$ is generated. Finally, we show that the applicability of region-based STC for perturbed/uncertain systems is semiglobal.

\subsection{Approximations of Isochronous Manifolds of Perturbed/Uncertain ETC Systems}\label{approximations section}
Similarly to \cite{delimpaltadakis_tac}, we upper-bound the time evolution of the (homogenized) triggering function $\tilde{\phi}(\xi(t;x),w(t))$ along the trajectories of DI \eqref{main_diff_incl} with a function $\mu\Big((x,w),t\Big)$ in analytic form that satisfies \eqref{bound requirements}. For this purpose, first we provide a lemma, similar to the comparison lemma \cite{khalil1996noninear} and to Lemma V.2 from \cite{delimpaltadakis_tac}, that shows how to derive upper-bounds with linear dynamics of functions evolving along flowpipes of differential inclusions: 
\begin{lemma}\label{bounding_lemma}
	Consider a system of ODEs: 
	\begin{equation}\label{ode_lemma}
		\dot{\xi}(t)=f(\xi(t),d(t)),
	\end{equation}
	where $\xi(t)\in\real^n$, $d(t)\in\real^{m_d}$, $f:\real^n\times\real^{m_d}\to\real^n$ and the function  $\phi:\real^n\to\real$. Let $f$, $d$ and $\phi$ satisfy Assumption \ref{assum1}. Consider the DI abstracting the family of ODEs \eqref{ode_lemma}:
	\begin{equation}\label{di_lemma}
		\dot{\xi}(t) \in F(\xi(t)) := \{f(\xi(t),d(t)):\text{ }d(t)\in\Delta\}.
	\end{equation}
	Consider a compact set $\Xi \subseteq \real^n$. For coefficients $\delta_{0}, \delta_{1}\in\real$ satisfying:
	\begin{equation}\label{delta inequality}
	\frac{\partial \phi}{\partial z}(z)f(z,u) \leq \delta_0\phi(z) + \delta_1, \quad \forall z \in \Xi \text{ and }\forall u\in\Delta,
	\end{equation}
	the following inequality holds for all $\xi_0 \in \Xi$:
	\begin{equation*}
	\sup\Big\{\phi\Big(\reach_t^F(\xi_0)\Big)\Big\} \leq \psi(y(\xi_0),t) \quad \forall t \in [0,t_e(\xi_0)],
	\end{equation*}
	where $t_e(\xi_0)$ is defined as the escape time: 
	\begin{equation}
	t_e(\xi_0)=\inf\{t>0:\reach_t^F(\xi_0) \not\subseteq\Xi\},
	\label{escape time}
	\end{equation}
	and $\psi(y(\xi_0),t)$ is:
	\begin{equation} \label{psi1}
	\psi(y(\xi_0),t) = \begin{bmatrix}1&0\end{bmatrix}\boldsymbol{e}^{At}y(\xi_0),
	\end{equation}
	where: 
	\begin{equation}\label{A_matrix}
	A = \begin{bmatrix}
	\delta_0 &1\\ 0 &0
	\end{bmatrix},\quad y(\xi_0) = \begin{bmatrix}
	\phi(\xi_0)\\ \delta_1
	\end{bmatrix}.
	\end{equation}
\end{lemma}
\begin{proof}
	See Appendix.
\end{proof}
Observe that, in contrast to Lemma V.2 from \cite{delimpaltadakis_tac} where the coefficients $\delta_i$ need to be positive, here $\delta_i\in\real$. This is because here, due to lack of knowledge on the derivative (or even on the differentiability) of the unknown signal $d(t)$, we consider only the first-order time-derivative of $\phi$ (first-order comparison), while in \cite{delimpaltadakis_tac} higher-order derivatives of $\phi$ are considered (higher-order comparison). For more information on the higher-order comparison lemma, the reader is referred to \cite{delimpaltadakis_tac} and the references therein. 

Now, we employ Lemma \ref{bounding_lemma}, in order to construct an upper-bound $\mu\Big((x,w),t\Big)$ of the triggering function $\tilde{\phi}(\xi(t;x),w(t))$ that satisfies the conditions \eqref{bound requirements} (in the subspace $w>0$), which in turn implies that the zero-level sets of $\mu\Big((x,w),t\Big)$ are inner-approximations of isochronous manifolds of DI \eqref{main_diff_incl} and satisfy the properties mentioned in Proposition \ref{manifold properties}. First, consider a compact connected set $\setz\subset\real^n$ with $0\in\mathrm{int}(\setz)$, and the set $\setw =[\underline{w},\overline{w}]$, where $\overline{w}>\underline{w}>0$. Define the following sets:
\begin{equation}\label{sets}
	\begin{aligned}
		&\Phi :=\bigcup\limits_{x_0\in\setz} \{x\in\real^n: e=x_0-x, w\in\setw, \tilde{\phi}\Big((x,e,w)\Big)\leq 0\},\\
		&\sete := \{x_0-x\in\real^n:x_0\in\setz,\text{ }x\in\Phi\},\\
		&\Xi := \Phi\times\sete\times\setw.
	\end{aligned}
\end{equation}
For the remaining, we assume the following:
\begin{assumption} \label{assum2}
	The set $\Phi\subset\real^n$ is compact.
\end{assumption}
Assumption \ref{assum2} is is satisfied by most triggering functions $\phi(\cdot)$ in the literature (e.g. Lebesgue sampling and most cases of Mixed Triggering from Remark \ref{trig_functions_remark}, the triggering functions of \cite{tabuada2007etc,girard2015dynamicetc}, etc.). Moreover, since $\Phi$ is assumed compact, then $\sete$ is compact as well, which implies that $\Xi$ is compact.  
\begin{remark}
	As it is discussed after Theorem \ref{main theorem}, the sets $\setz, \setw,\Phi,\sete,\Xi$ are constructed such that for all initial conditions $(x,0,w)\in\setz\times\sete\times\setw$, the trajectories of DI \eqref{main_diff_incl} reach the boundary of $\Xi$ after (or at) the inter-sampling time $t=\tau\Big((x,w)\Big)$. An alternative construction of such sets 
	has been proposed in \cite{tosample,delimpaltadakis_tac} and utilizes a given Lyapunov function for system \eqref{perturbed_etc_sys} and its level sets. 
\end{remark}

The following theorem shows how the bound $\mu\Big((x,w),t\Big)$ is constructed:
\begin{theorem}\label{main theorem}
	Consider the family of ETC systems \eqref{perturbed_etc_sys}, the DI \eqref{main_diff_incl} abstracting them, a homogenized triggering function $\tilde{\phi}(\xi(t;x),w(t))$, the sets $\setz, \setw,\Phi,\sete,\Xi$ defined in \eqref{sets} and let Assumptions \ref{assum1} and \ref{assum2} hold. Let $\delta_0\geq0$ and $\delta_1>0$ be such that:
	\begin{subequations}\label{delta_requirements}
		\begin{align}
		&\forall(z,w,u)\in\Xi\times\Delta:
		\qquad\frac{\partial \tilde{\phi}}{\partial z}(z,w)w^{\alpha+1}f_e(w^{-1}z,u) \leq \delta_0\tilde{\phi}(z,w) + \delta_1, \label{delta_ineq_theorem}\\
		&\forall(z,w)\in\setz\times\{0\}\times\setw:\text{ }\delta_0\tilde{\phi}(z,w)+\delta_1\geq\varepsilon_\delta>0,\label{delta_init_cond}
		\end{align}
	\end{subequations}
	where $\varepsilon_\delta$ an arbitrary positive constant. Let $r>\underline{w}$ be such that $D_r:=\{(x,w)\in\real^{n+1}:|(x,w)|=r, \text{ }w\in\setw\}\subset\setz\times\setw$. For all $(x,w)\in\real^{n+1}\setminus\{0\}$ define the function:
	\begin{equation}\label{mu}
		\mu\Big((x,w),t\Big) := \Big(\tfrac{|(x,w)|}{r}\Big)^{\theta+1}\begin{bmatrix}1&0\end{bmatrix}\boldsymbol{e}^{A\Big(\tfrac{|(x,w)|}{r}\Big)^\alpha t}y(x,w),
	\end{equation}
	where $A$ is as in \eqref{A_matrix} and:
	\begin{equation*}
		y(x,w) = \begin{bmatrix}
		\tilde{\phi}\Big((r\tfrac{x}{|(x,w)|},0,r\tfrac{w}{|(x,w)|})\Big)\\
		\delta_1
		\end{bmatrix}.
	\end{equation*}
	The function $\mu\Big((x,w),t\Big)$ satisfies \eqref{bound init cond}, \eqref{scaling of bound}, \eqref{one-zero-crossing of bound} for all $(x,w)\in(\real^n\times\real_+)\setminus\{0\}$, but condition \eqref{bound time validity} is satisfied only in the cone
	\begin{equation}\label{cone}
		\cone=\{(x,w)\in\real^n\times\real_+:|x|^2+w^2\leq\frac{w^2}{\underline{w}^2}r^2\}\setminus\{0\}
	\end{equation}
	and $\forall t\in[0,\tau\Big((x,w)\Big)]$.
\end{theorem}
\begin{proof}
	See Appendix.
\end{proof}
\begin{remark} \label{delta_existence_remark}
	Observe that, under Assumptions \ref{assum1} and \ref{assum2}, the term $\frac{\partial \tilde{\phi}}{\partial z}(z,w)w^{\alpha+1}f_e(w^{-1}z,u)$ is bounded for all $(z,w,u)\in\Xi\times\Delta$, since $f_e$ is locally bounded, $\phi$ is continuously differentiable (implying that $\tilde{\phi}$ is also continuously differentiable for $w\neq0$), $\tilde{\phi}(z,w)$ is bounded for all $(z,w)\in\setz\times\{0\}\times\setw$ and $\Xi\times\Delta$ is compact and does not contain any point $(z,0,u)$. Thus, coefficients $\delta_0\geq0$ and $\delta_1>0$ satisfying \eqref{delta_requirements} always exist; e.g. $\delta_0=0$ and $\delta_1 > \max\Big\{\epsilon_\delta,\sup\limits_{(z,w,u)\in\Xi\times\Delta}\frac{\partial \tilde{\phi}}{\partial z}(z,w)w^{\alpha+1}f_e(w^{-1}z,u)\Big\}$. In \cite{delimpaltadakis_tac}, a computational algorithm has been proposed, which computes the coefficients $\delta_i$ for a given ETC system and triggering function, by employing Linear Programming and Satisfiability-Modulo Theory solvers (SMT, see e.g. \cite{dreal}).
\end{remark}
Let us explain the intuition behind Theorem \ref{main theorem}. First, observe that, according to Lemma \ref{bounding_lemma}, the coefficients $\delta_0,\delta_1$ satisfying \eqref{delta_ineq_theorem}, determine a function $\psi(y((x,0,w_\star)),t)$ that upper bounds $\sup\Big\{\phi\Big(\reach_t^{\tilde{F}}((x,0,w_\star))\Big)\Big\}$. The sets $\setz,\setw,\Phi,\sete,\Xi$ have been chosen such that the inequality\begin{equation*}
\psi(y((x,0,w_\star)),t)\geq\sup\Big\{\phi\Big(\reach_t^{\tilde{F}}((x,0,w_\star))\Big)\Big\}
\end{equation*}  holds for all $t\in[0,\tau\Big((x,w_\star)\Big)]$. Now, introducing the scaling terms $\Big(\tfrac{|(x,w)|}{r}\Big)^\alpha,r\tfrac{x}{|(x,w)|},$ etc. which projects $\psi(\cdot)$ onto the spherical segment $D_r$ and transforms it into $\mu(\cdot)$, enforces that $\mu(\cdot)$ satisfies \eqref{bound time validity} and the scaling property \eqref{scaling of bound}. Inequalities $\delta_0\geq0,\delta_1>0$ and \eqref{delta_init_cond} enforce that $\mu(\cdot)$ satisfies \eqref{one-zero-crossing of bound}. Finally, \eqref{bound time validity} being satisfied only in the cone $\cone$, stems from the fact that $0\notin\mathrm{int}(\setw)$. Note that $\setw$ is chosen such that it is guaranteed that \eqref{delta_requirements} is well-defined everywhere in $\Xi\times\Delta$.

The fact that \eqref{bound time validity} is satisfied only in the cone $\cone$ has the following implication:
\begin{corollary}[to Theorem \ref{theorem1_tac}]\label{cone corollary}
	Consider the family of ETC systems \eqref{perturbed_etc_sys}, the DI \eqref{main_diff_incl} abstracting them, a (homogenized) triggering function $\tilde{\phi}(\xi(t;x),w(t))$ and let Assumptions \ref{assum1} and \ref{assum2} hold. Consider the function $\mu\Big((x,w),t\Big)$ from \eqref{mu}. The sets $\underline{M}_{\tau_{\star}} = \{(x,w)\in\mathbb{R}^{n+1}: \mu\Big((x,w),\tau_{\star}\Big)=0\}$ inner-approximate isochronous manifolds $M_{\tau_{\star}}$ of DI \eqref{main_diff_incl} inside the cone $\cone$, i.e. for all $(x,w)\in\underline{M}_{\tau_{\star}}\cap\cone$:
	\begin{equation*}
	\begin{aligned}
		&\bullet\exists! \kappa_{(x,w)}\geq 1 \text{ s.t. } \kappa_{(x,w)}(x,w) \in M_{\tau_\star}\\
		&\bullet \not\exists \lambda_{(x,w)}\in(0,1) \text{ s.t. } \lambda_{(x,w)}(x,w) \in M_{\tau_\star}.
	\end{aligned}
	\end{equation*}
	Moreover, the sets $\underline{M}_{\tau_{\star}}$ satisfy the properties mentioned in Proposition \ref{manifold properties}.
\end{corollary}
\begin{proof}
	It follows identical arguments to the proof of Theorem \ref{theorem1_tac} in \cite{delimpaltadakis_tac}. The only difference is that the arguments are now made for all $(x,w)\in \cone$ and not for all $(x,w)\in\real^{n+1}$.
\end{proof}
The implications of the above corollary are depicted in Figure \ref{semiglobal_approx_fig}. 
According to Section \ref{inner_approx_section}, since the zero-level sets $\underline{M}_{\tau_i}$ of $\mu\Big((x,w),t\Big)$ inner-approximate isochronous manifolds inside $\cone$, for the regions $\reg_i$ that are delimited by consecutive approximations $\underline{M}_{\tau_i}$ and the cone $\cone$ (see Figure \ref{semiglobal_approx_fig}) it holds that: $\tau_i\leq\tau\Big((x,w)\Big)$.
\begin{figure}[!h]
	\centering
	\includegraphics[width=2in]{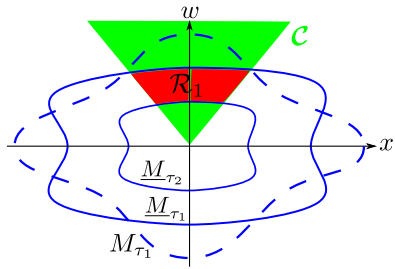}
	\caption{Isochronous manifold $M_{\tau_1}$ (solid line) and approximations of isochronous manifolds $\underline{M}_{\tau_1},\underline{M}_{\tau_2}$ (dashed lines). The set $\underline{M}_{\tau_1}$ inner-approximates $M_{\tau_1}$ only inside the cone $\cone$. The red region $\reg_1$ contained between $\underline{M}_{\tau_1},\underline{M}_{\tau_2}$ and the cone $\cone$ satisfies \eqref{regions-times}.}
	\label{semiglobal_approx_fig}
\end{figure} 
Thus, given the set of times $\{\tau_1,\dots,\tau_q\}$, the regions $\reg_i$ are defined as the regions between consecutive approximations $\underline{M}_{\tau_i}$ and the cone $\cone$:
\begin{equation}\label{regions}
\begin{aligned}
\reg_i := \Big\{(x,w)\in\cone:\text{ } &\mu\Big((x,w),\tau_i\Big)\leq 0,\\&\mu\Big((x,w),\tau_{i+1}\Big)\geq 0 \Big\}.
\end{aligned}	
\end{equation}
As discussed in Section \ref{tac_homogenization}, in a real-time implementation, given a measurement $x\in\real^n$, the controller checks to which region $\reg_i$ the point $(x,1)\in\real^{n+1}$ belongs, and correspondingly decides the next sampling time instant (see Figure \ref{w_1_discr}).
\begin{remark}
	The innermost region $\reg_q$ cannot be defined as in \eqref{regions}, as there is no $\tau_{q+1}$. For $\reg_q$, it suffices that we write:
	\begin{equation*}
	\reg_q := \Big\{(x,w)\in\cone: \text{ }\mu\Big((x,w),\tau_q\Big)\leq 0\Big\}
	\end{equation*}
\end{remark}

\subsection{Semiglobal Nature of Region-Based STC}
It is obvious that the regions $\reg_i$ do not cover the whole $w=1$-hyperplane (which is where the state space of the original system is mapped), i.e. there exist states $x\in\real^n$ such that the point $(x,1)\in\real^{n+1}$ does not belong to any region $\reg_i$, and thus no STC inter-sampling time can be assigned to $x$. Let us demonstrate which set $\mathcal{B}\subseteq\real^n$ is covered by the partition created and show that it can be made arbitrarily large. 

The set $\mathcal{B}$ is composed of all points $x\in\real^n$ such that $(x,1)$ belongs to any region $\reg_i$, i.e.:
\begin{equation*}
\mathcal{B}:=\{x\in\real^n:(x,1)\in\bigcup\limits_i\reg_i\}.
\end{equation*}
From the definition \eqref{regions} of regions $\reg_i$ and the scaling property \eqref{scaling of bound} of $\mu(\cdot)$, it follows that $\bigcup\limits_i\reg_i=\cone\cap\{(x,w)\in\real^n\times\real_+:\mu\Big((x,w),\tau_1\Big)\leq 0\}$. By fixing $w=1$ in the expression \eqref{cone} of $\cone$ and in $\{(x,w)\in\real^n\times\real_+:\mu\Big((x,w),\tau_1\Big)\leq 0\}$, we get:
\begin{align}
&\bullet (x,1)\in\cone \iff x\in \{x\in\real^n:|x|^2\leq\frac{r^2-\underline{w}^2}{\underline{w}^2}\}=:B_1,\label{B1}\\
&\bullet (x,1)\in\{(x,w)\in\real^n\times\real_+:\mu\Big((x,w),\tau_1\Big)\leq 0\} \iff\nonumber\\& x\in\{x\in\real^n:\mu\Big((x,1),\tau_1\Big)\leq 0\}=:B_2
\end{align}
Thus, we can write the set $\mathcal{B}$ as:
\begin{equation}\label{ball}
\mathcal{B}:=\{x\in\real^n:x\in B_1,x\in B_2\} = B_1\cap B_2.
\end{equation}
The set $B_1$ is depicted in Figure \ref{semiglobality_fig} in the Appendix. Since $r>\underline{w}$, $B_1$ is non-empty. Moreover, we can choose $\underline{w}>0$ to be arbitrarily small without changing $r$, therefore we can make the set $B_1$ arbitrarily large. Finally, $B_2$ is non-empty (as it is the set delimited by $\underline{M}_{\tau_1}$ and $\cone$) and, owing to the scaling property \eqref{scaling of bound} of $\mu(\cdot)$, it can be made arbitrarily large by selecting a sufficiently small $\tau_1$. Consequently, $\mathcal{B}$ is non-empty, and can be made arbitrarily large. Hence, region-based STC is applicable semiglobally in $\real^n$.
\begin{remark}\label{remark_number_of_regions}
	As discussed in \cite{delimpaltadakis_tac}, for fixed $\tau_1$ and $\tau_q$, as the total number $q$ of predefined times $\tau_i$ grows, the sets $\reg_i$ become smaller (since the same set $\mathcal{B}$ is partitioned into more regions $\reg_i$). This increases the accuracy of times $\tau_i$ as lower bounds of the actual ETC times $\tau(x)$, but also increases the on-line computational load of the controller
	, thus providing a trade-off between performance and computations.
\end{remark}

\section{Numerical Example}\label{numerical example section}
Let us demonstrate how the proposed STC is applied to a perturbed uncertain system, and compare its performance to the STC of \cite{small_gain_robust_etc}.  Consider the ETC system from \cite{small_gain_robust_etc}:
\begin{equation} \label{example_1}
	\dot{\zeta}_1 = \zeta_2 + g_1(\zeta_1,d_1), \quad \dot{\zeta}_2 = u(\zeta,\varepsilon_\zeta) + g_2(\zeta_2),
\end{equation}
where $|g_1(\zeta_1,d_1)|\leq0.1|\zeta_1|+0.1|d_1|$ and $|g_2(\zeta_2)|\leq0.2|\zeta_2|^2$ are uncertain, and $d_1(t)$ is an unknown bounded disturbance with $|d_1(t)|\leq4$. The ETC feedback $u$ is $u(\zeta,\varepsilon_\zeta)=-(7.02|\zeta_2+\varepsilon_{\zeta_2}-p_1|-25.515)(\zeta_2+\varepsilon_{\zeta_2}-p_1)$, where $p_1 = -2.1(\zeta_1+\varepsilon_{\zeta_1})$. The triggering function from \cite{small_gain_robust_etc}, that is to be emulated, is:
\begin{equation}\label{example_trig_fun}
	\phi(\zeta,\varepsilon_{\zeta})=|\varepsilon_{\zeta}(t)|^2-0.0049|\zeta(t)|^2-16,
\end{equation}
which guarantees convergence to a ball (practical stability). First, we bring \eqref{example_1} to the form of \eqref{perturbed_etc_sys}, by writing:
\begin{equation} \label{example_2}
\dot{\xi}(t)=\begin{bmatrix}
\dot{\zeta}_1\\\dot{\zeta}_2\\\dot{\varepsilon}_{\zeta_1}\\\dot{\varepsilon}_{\zeta_2}
\end{bmatrix}=\begin{bmatrix}
\zeta_2 + 0.1d_2\zeta_1+0.1d_1\\u(\zeta,\varepsilon_\zeta) + 0.2d_3\zeta_2^2\\-\zeta_2  -0.1d_2\zeta_1-0.1d_1\\-u(\zeta,\varepsilon_\zeta) - 0.2d_3\zeta_2^2
\end{bmatrix} = f_e(\xi(t),d(t))
\end{equation}
where $d(t)=(d_1(t),d_2(t),d_3(t))\in[-4,4]\times[-1,1]^2$, i.e. $\Delta = [-4,4]\times[-1,1]^2$. Observe that Assumption \ref{assum1} is satisfied. Then, we construct the homogeneous DI abstracting \eqref{example_2} according to \eqref{main_diff_incl}:
\begin{equation}\label{example_hom_di}
\begin{pmatrix}
\dot{\xi}(t)\\\dot{w}(t)
\end{pmatrix} = \begin{bmatrix}
\{w^{2}f_e(w^{-1}\xi,d(t)):\text{ }d(t)\in\Delta\}\\
0
\end{bmatrix},
\end{equation}
and homogenize the triggering function as follows:
\begin{equation}\label{example_hom_fun}
\tilde{\phi}(\xi(t),w(t)) = |\varepsilon_{\zeta}(t)|^2-0.0049|\zeta(t)|^2-16w^2(t).
\end{equation}
The degree of homogeneity for both \eqref{example_hom_di} and \eqref{example_hom_fun} is 1. 

Next, we derive the $\delta_i$ coefficients according to Theorem \ref{main theorem}, to determine the regions $\reg_i$. We fix $\setz=[-0.1,0.1]^2$, $\setw = [10^{-6},0.1]$ and define the sets $\Phi,\sete,\Xi$ as in \eqref{sets}, where $\Phi$ is indeed compact. By employing the computational algorithm of \cite{delimpaltadakis_tac}, $\delta_0\approx0.0353$ and $\delta_1\approx0.3440$ are obtained. We choose $r=0.099$ such that $D_r\subset\setz\times\setw$, and define $\mu\Big((x,w),t\Big)$ as in \eqref{mu}. Finally, the state-space of DI \eqref{example_hom_di} is partitioned into 434 regions $\reg_i$ with $\tau_1 \approx 63\cdot10^{-5}$ and $\tau_{i+1}=1.01\tau_i$.

We ran a number of simulations to compare our approach to the approach of \cite{small_gain_robust_etc} and to the ideal performance of the emulated ETC \eqref{example_trig_fun}. More specifically, we simulated the system for 100 different initial conditions uniformly distributed in a ball of radius 2. The simulations' duration is $5$s. As in \cite{small_gain_robust_etc}, we fix: $g_1(\zeta_1,d_1)=0.1\zeta_1\sin(\zeta_1)+0.1d_1$, $d_1=4\sin(2\pi t)$ and $g_2(\zeta_2)=0.2\zeta_2^2\sin(\zeta_2)$. The self-triggered sampler of \cite{small_gain_robust_etc} determines sampling times as follows: $t_{i+1}=t_i + \frac{1.54}{28(|x_i|+4)+29}$, where $x_i$ is the state measured at $t_i$. 
\begin{figure}[!h]
	\centering
	\includegraphics[width=3.5in]{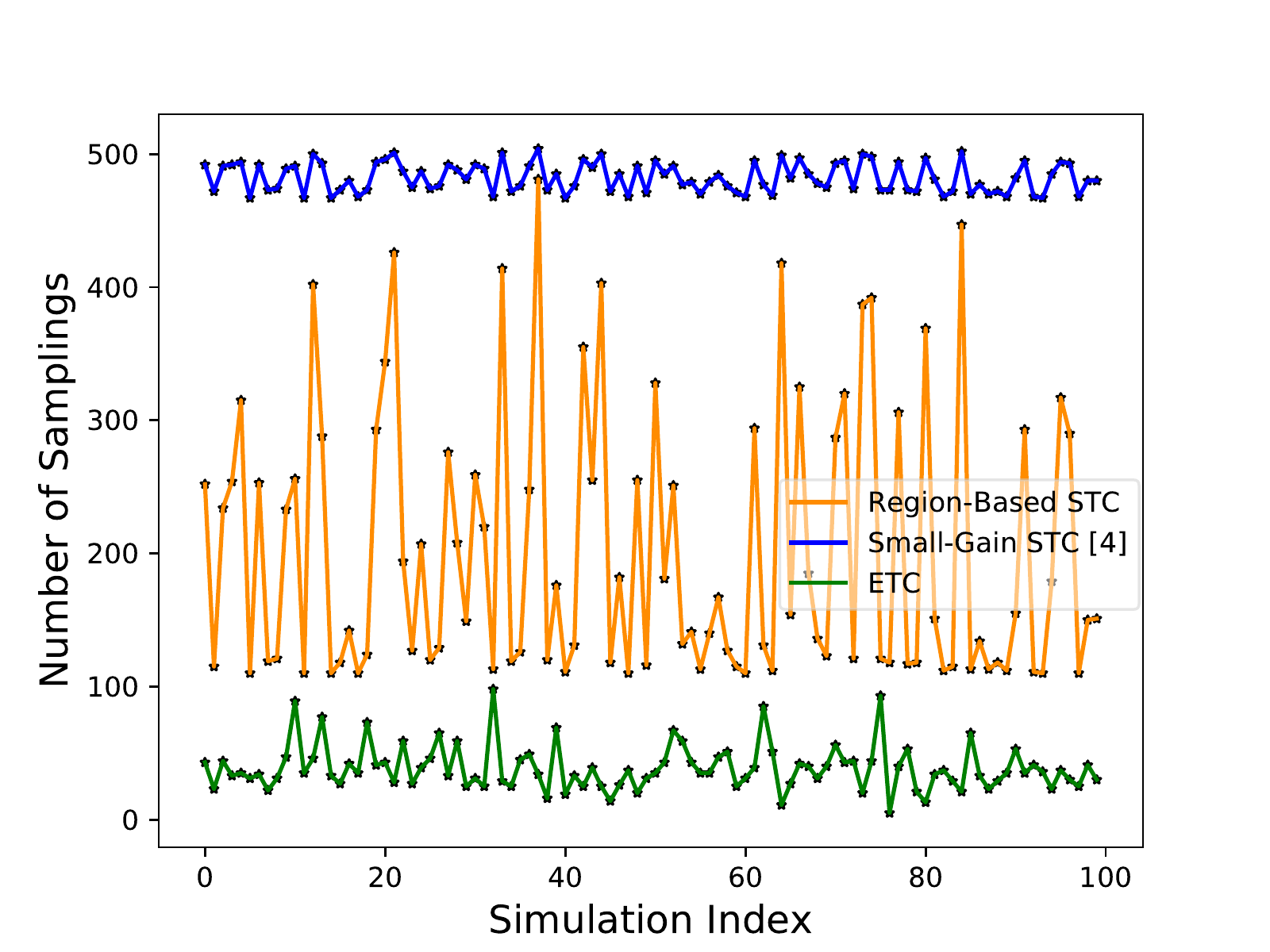}
	\caption{Number of samplings for each simulation of region-based STC (orange), STC of \cite{small_gain_robust_etc} (blue) and ETC \eqref{example_trig_fun} (green).}
	\label{comparison_fig}
\end{figure} 
The total number of samplings for each simulation of all three schemes is depicted in Fig. \ref{comparison_fig}. The average number of samplings per simulation was: 200.71 for region-based STC, 482.32 for STC \cite{small_gain_robust_etc} and 38.81 for ETC. We observe that region-based STC is in general less conservative than the STC of \cite{small_gain_robust_etc}, while being more versatile as well. Recall that the main advantage of our approach is its versatility compared to the rest of the approaches, in terms of its ability to handle different performance specifications and different types of system's dynamics, provided that an appropriate triggering function is given. For example, \cite{small_gain_robust_etc} is constrained to ISS systems, while our approach does not obey such a restriction. Finally, as expected, ETC leads to a smaller amount of samplings compared to both STC schemes. 

We, also, present illustrative results for one particular simulation with initial condition $(-1,-1)$. Figure \ref{traj_fig} shows the trajectories of the system when controlled via region-based STC and the STC from \cite{small_gain_robust_etc}, while Figure \ref{intersampling_times_fig} shows the time-evolution of inter-sampling times for the two schemes. Region-based STC led to 166 samplings, whereas the STC of \cite{small_gain_robust_etc} led to 483. We observe that, while the performance of both schemes is the same (the trajectories are almost identical in Figure \ref{traj_fig}), region-based STC leads to a smaller amount of samplings, i.e. less resource utilization. Moreover, from Figure \ref{intersampling_times_fig} we observe that, especially during the steady-state response, region-based STC performs considerably better, in terms of sampling. However, there is a small period of time in the beginning of the simulation, when the trajectories overshoot far away from the origin and region-based STC gives faster sampling. Finally, we have to note that while we have not added more comparative simulations with the other STC schemes that address disturbances or uncertainties \cite{tiberi_stc_perturbed,med_stc,italy_digital_stc} for conciseness, simulation results have indicated that region-based STC is competitive to these approaches as well.
	\begin{figure}[h]
		\centering
		\includegraphics[width=2.5in]{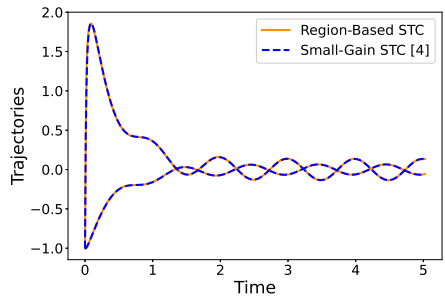}
		\caption{Trajectories of system \eqref{example_1} with initial condition $(-1,-1)$, under region-based STC (orange lines) and the STC of \cite{small_gain_robust_etc} (dashed blue lines).}
		\label{traj_fig}
	\end{figure}\quad
	\begin{figure}[h]
		\centering
		\includegraphics[width=2.5in]{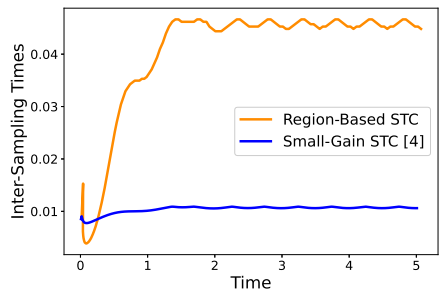}
		\caption{Evolution of inter-sampling times during a simulation with initial condition $(-1,-1)$, for region-based STC (orange line) and the STC of \cite{small_gain_robust_etc} (blue line).}
		\label{intersampling_times_fig}
	\end{figure}

\section{Conclusions}
In this work, by extending the work of \cite{delimpaltadakis_tac}, we have proposed a region-based STC scheme for nonlinear systems with disturbances and uncertainties, that is able to provide different performance guarantees, depending on the triggering function that is chosen to be emulated. By employing a framework based on DIs and introducing ETC notions therein, we have extended significant results on ETC/STC to perturbed uncertain systems. Employing the renewed results, we have constructed approximations of isochronous manifolds of perturbed/uncertain systems, enabling region-based STC. The provided numerical simulations indicate that our approach, while being more versatile, is competitive with respect to other approaches as well, in terms of inter-sampling times. It is worth noting that region-based STC provides room for numerous extensions, due to the generic way of converting ETC to STC it offers. For example, (dynamic) output-feedback could easily be considered by incorporating the (controller's and) observer's dynamics into the system description. Hence, for future work, we will consider several extensions of the proposed STC (e.g. to systems with communication delays). Apart from that, we plan on utilizing the derived approximations of isochronous manifolds to construct timing models of perturbed uncertain nonlinear ETC systems for traffic scheduling in networks of ETC loops, building upon \cite{delimpa2020traffic}.

\appendix
\begin{proof}[Proof of Theorem \ref{di_intersampling_scaling_theorem}]
	According to the definition of inter-sampling times \eqref{intersampling_times_di}, for $\tau\Big(\lambda (x,w)\Big)$ we have:
	\begin{equation*}
	\tau\Big(\lambda (x,w)\Big) = \inf\bigg\{t>0:\sup\Big\{\tilde{\phi}\Big(\reach^{\tilde{F}}_t(\lambda(x,0,w))\Big)\Big\}\geq 0\bigg\}
	\end{equation*}
	Employing the scaling property \eqref{scaling_di} and the fact that $\tilde{\phi}$ is homogeneous of degree $\theta$, we can write $\tau\Big(\lambda (x,w)\Big)$ as:
	\begin{align*}
	&\inf\bigg\{t>0:\sup\Big\{\tilde{\phi}\Big(\lambda\reach^{\tilde{F}}_{\lambda^\alpha t}((x,0,w))\Big)\Big\}\geq 0\bigg\}=\\
	&\inf\bigg\{t>0:\sup\Big\{\lambda^{\theta+1}\tilde{\phi}\Big(\reach^{\tilde{F}}_{\lambda^\alpha t}((x,0,w))\Big)\Big\}\geq 0\bigg\}=\\
	&\inf\bigg\{\lambda^{-\alpha}t>0:\sup\Big\{\tilde{\phi}\Big(\reach^{\tilde{F}}_{ t}((x,0,w))\Big)\Big\}\geq 0\bigg\}=\\
	&\lambda^{-\alpha}\tau\Big((x,w)\Big)
	\end{align*}
\end{proof}
\begin{proof}[Proof of Lemma \ref{bounding_lemma}]
	Consider the restriction of ODE \eqref{ode_lemma} to the set $\Xi$:
	\begin{equation}\label{lemma_proof_ode}
		\dot{\xi}(t) = f(\xi(t),d(t)), \quad \xi(t)\in\Xi.
	\end{equation}	
	Any solution of \eqref{lemma_proof_ode} is also a solution of \eqref{ode_lemma} (possibly not a maximal one). Note that \eqref{delta inequality} is equivalent to:
	\begin{equation}\label{lemma_proof_eq1}
		\dot{\phi}(\xi(t;\xi_0))\leq\delta_0\phi(\xi(t;\xi_0))+\delta_1,
	\end{equation}
	where $\xi(t;\xi_0)$ is any solution of \eqref{lemma_proof_ode}, with $\xi_0\in\Xi$. Observe that $\psi(y(\xi_0),t)$ is the solution to the scalar differential equation $\dot{\psi}=\delta_{0}\psi+\delta_1$ with initial condition $\psi_0=\phi(\xi_0)$:
	\begin{equation*}
		\psi(y(\xi_0),t)=\begin{bmatrix}1&0\end{bmatrix}\boldsymbol{e}^{At}y(\xi_0)
		=e^{\delta_0t}\phi(\xi_0)+\frac{e^{\delta_0t}-1}{\delta_0}\delta_1.
	\end{equation*}
	Thus, by employing the comparison lemma (see \cite{khalil1996noninear}, pp. 102-103), from \eqref{lemma_proof_eq1} we get that for any $d_{\star}(t)$ satisfying Assumption \ref{assum1} and all $\xi_0\in\Xi$:
	\begin{equation}\label{lemma_proof_eq2}
		\phi(\xi(t;\xi_0))\leq\psi(y(\xi_0),t), \quad \forall t \in[0,t_{e,d_{\star}}(\xi_0)),
	\end{equation}
	where $[0,t_{e,d_{\star}}(\xi_0))$ is the maximal interval of existence of solution $\xi(t;\xi_0)$ to ODE \eqref{lemma_proof_ode} under the realization $d(t)=d_{\star}(t)$. The time $t_{e,d_{\star}}(\xi_0)$ is defined as the time when $\xi(t;\xi_0)$, under the realization $d(t)=d_{\star}(t)$, leaves the set $\Xi$: 
	\begin{equation*}
	\begin{aligned}\small
		t_{e,d_{\star}}(\xi_0) = \sup\{\tau>0:d(t)=d_{\star}(t),\xi(t;\xi_0) \in \Xi\text{ } \forall t \in [0,\tau)\}
	\end{aligned}	
	\end{equation*}
	Since \eqref{lemma_proof_eq2} holds for all $d_{\star}(t)$ satisfying Assumption \ref{assum1}, we can conclude that $\psi(y(\xi_0),t)$ 
	bounds all solutions of DI \eqref{di_lemma} starting from $\xi_0\in\Xi$ as follows:
	 \begin{equation*}
	 	\sup\Big\{\phi\Big(\reach_t^F(\xi_0)\Big)\Big\} \leq \psi(y(\xi_0),t), \quad \forall t \in [0,\inf\limits_{d_\star}t_{e,d_{\star}}(\xi_0)).
	 \end{equation*}
	 Finally, note that $\inf\limits_{d_\star}t_{e,d_{\star}}(\xi_0)$ represents the smallest possible $\Xi$-escape time among all trajectories generated by DI \eqref{di_lemma}, i.e. $\inf\limits_{d_\star}t_{e,d_{\star}}(\xi_0) = \inf\{t>0:\reach_t^F(\xi_0) \not\subseteq\Xi\} = t_e(\xi_0)$. Hence, we can conclude that:
	 \begin{equation*}
	 	\sup\Big\{\phi\Big(\reach_t^F(\xi_0)\Big)\Big\} \leq \psi(y(\xi_0),t), \quad \forall t \in [0,t_e(\xi_0)).
	 \end{equation*} 
\end{proof}
\begin{proof}[Proof of Theorem \ref{main theorem}]
	First notice that, under item \ref{trig_fun_negativity} of Assumption \ref{assum1}, \eqref{bound init cond} holds: $\mu\Big((x,w),0\Big)=\Big(\tfrac{|(x,w)|}{r}\Big)^{\theta+1}\tilde{\phi}\Big((r\tfrac{x}{|(x,w)|},0,r\tfrac{w}{|(x,w)|})\Big)<0$ for all $(x,w)\in\real^{n+1}\setminus\{0\}$. Moreover, observe that $\mu(\cdot,\cdot)$ satisfies the time-scaling property \eqref{scaling of bound} by construction. It remains to prove that $\mu(\cdot,\cdot)$ satisfies \eqref{bound time validity} and \eqref{one-zero-crossing of bound}.	
	
	In order to prove that $\mu(\cdot,\cdot)$ satisfies \eqref{bound time validity}, as already explained in Section \ref{approximations section}, we follow the following steps: 1) we show that the coefficients $\delta_0,\delta_1$ satisfying \eqref{delta_ineq_theorem} determine a function $\psi(y((x,0,w_\star)),t)$ satisfying \eqref{theorem_proof_eq1}, 2) using the sets $\setz,\setw,\sete,\Phi,\Xi$ we show that $\psi(y((x,0,w_\star)),t)$ satisfies \eqref{theorem_proof_eq2}, and finally 3) observing that $\mu$ is obtained by a projection of $\psi$ to $D_r$, we show that $\mu$ satisfies \eqref{bound time validity} (see \eqref{theorem_proof_eq6}). 
	
	Let us formally prove it. 
	Assumption \ref{assum1} implies that $\tilde{F}(\xi,w)\subseteq\real^{2n+1}$ is non-empty, compact and convex for any $(\xi,w)\in\real^{2n+1}\setminus\{0\}$ and outer-semicontinuous. These conditions ensure existence and extendability of solutions for each initial condition \cite{filippov2013differential}. According to Lemma \ref{bounding_lemma} and since $\Xi$ is compact, the coefficients $\delta_0,\delta_1$ satisfying \eqref{delta_ineq_theorem}, determine a function $\psi(y((x,e,w_\star)),t)$ such that for all $(x,e,w)\in\Xi$: $\psi(y((x,e,w)),t)\geq\sup\Big\{\phi\Big(\reach_t^{\tilde{F}}((x,e,w))\Big)\Big\}$, 
	$\forall t\in[0,t_e\Big((x,e,w)\Big)],$
	 where $t_e\Big((x,e,w)\Big)$ is defined in \eqref{escape time} as the time when $\reach_t^{\tilde{F}}((x,e,w))$ leaves the set $\Xi$. Since we are only interested in initial conditions with the measurement error component being 0, we write:
	 \begin{equation}\label{theorem_proof_eq1}
	 \begin{aligned}
	 \small
	 \psi(y((x,0,w)),t)&\geq\sup\Big\{\phi\Big(\reach_t^{\tilde{F}}((x,0,w))\Big)\Big\},\\
	 &\forall(x,0,w)\in\Xi\text{ and }\forall t\in[0,t_e\Big((x,0,w)\Big)].
	 \end{aligned}	
	 \end{equation}
	 Observe that for all initial conditions $(x,0,w)\in\setz\times\sete\times\setw$, the sets $\Phi$ and $\sete$ are exactly such that $\xi(t;(x,0))\notin\Phi\times\sete\implies\phi(\xi(t;(x,0)))> 0$, where $\xi(\cdot)$ represents the $\xi$-component of solutions of DI \eqref{main_diff_incl} (since $w(t)$ remains constant along solutions of DI \eqref{main_diff_incl}, we neglect it). Thus, all trajectories that start from any initial condition $(x,0,w)\in\setz\times\sete\times\setw$ reach the boundary of $\Xi = \Phi\times\sete\times\setw$ after (or at) the inter-sampling time $\tau\Big((x,w)\Big)$, i.e. $\tau\Big((x,w)\Big)\leq t_e\Big((x,e,w)\Big)$ for all $(x,w)\in\setz\times\setw$. Thus, employing \eqref{theorem_proof_eq1} we write:
	 \begin{equation}\label{theorem_proof_eq2}
	 \begin{aligned}
	 \psi(y((x,0,w)),t)&\geq\sup\Big\{\phi\Big(\reach_t^{\tilde{F}}((x,0,w))\Big)\Big\},\\
	 &\forall(x,w)\in\setz\times\setw\text{ and }\forall t\in[0,\tau\Big((x,w)\Big)].
	 \end{aligned}	
	 \end{equation}
	 Now, consider any point $(x_0,w_0)\in D_r\subseteq\setz\times\setw$. Observe that $\mu\Big((x_0,w_0),t\Big) = \psi(y((x_0,0,w_0)),t)$. Thus, since $D_r\subseteq\setz\times\setw$, from \eqref{theorem_proof_eq2} we get:
	 \begin{equation}\label{theorem_proof_eq3}
	 \begin{aligned}
	 \mu\Big((x_0,w_0),t\Big)&\geq\sup\Big\{\phi\Big(\reach_t^{\tilde{F}}((x_0,0,w_0))\Big)\Big\},\\
	 &\forall(x_0,w_0)\in D_r\text{ and }\forall t\in[0,\tau\Big((x_0,w_0)\Big)].
	 \end{aligned}	
	 \end{equation}
	 To prove that $\mu(\cdot)$ satisfies \eqref{bound time validity} in the cone $\cone$ from \eqref{cone}, we have to show that \eqref{theorem_proof_eq3} holds for all $(x,w)\in\cone$. First, observe that $\cone$ is defined as the cone stemming from the origin with its extreme vertices being all points in the intersection $D_r\cap\setz\times\setw$ (see Figure \ref{semiglobality_fig}).
	 \begin{figure}[!h]
	 	\centering
	 	\includegraphics[width=2.5in]{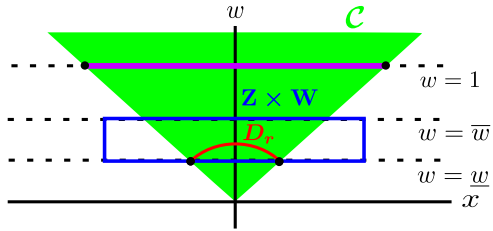}
	 	\caption{The sets $\setz\times\setw$ (region contained in blue box), $D_r$ (red spherical segment) and the cone $\cone$ (green) from \eqref{cone}. The subset of the hyperplane $w=1$ painted in purple represents the set $B_1$ from \eqref{B1}.}
	 	\label{semiglobality_fig}
	 \end{figure}
 	Thus, since $D_r$ is a spherical segment, for any point $(x,w)\in \cone$ there always exists a $\lambda>0$ and a point $(x_0,w_0)\in D_r$ such that $(x,w)=\lambda(x_0,w_0)$. If we interchange $(x_0,w_0)$ with $\lambda^{-1}(x,w)$ in \eqref{theorem_proof_eq3}, we get:
	\begin{equation}\label{theorem_proof_eq4}
	\begin{aligned}
	&\mu\Big(\lambda^{-1}(x,w),t\Big)\geq\sup\Big\{\phi\Big(\reach_t^{\tilde{F}}(\lambda^{-1}(x,0,w))\Big)\Big\},\\
	&\forall(x,w)\in \cone\text{ and }\forall t\in[0,\tau\Big(\lambda^{-1}(x,w)\Big)].
	\end{aligned}	
	\end{equation}
	But, from \eqref{scaling_di}, \eqref{scaling of bound} and Theorem \eqref{di_intersampling_scaling_theorem} we get:
	\begin{equation}\label{theorem_proof_eq5}\small
	\begin{aligned}
	& \bullet\sup\Big\{\phi\Big(\reach_t^{\tilde{F}}(\lambda^{-1}(x,0,w))\Big)\Big\} =\\ &\lambda^{-\theta-1}\sup\Big\{\phi\Big(\reach_{\lambda^{-\alpha}t}^{\tilde{F}}((x,0,w))\Big)\Big\}\\
	&\bullet\mu\Big(\lambda^{-1}(x,w),t\Big) = \lambda^{-\theta-1}\mu\Big((x,w),\lambda^{-\alpha}t\Big)\\
	&\bullet\tau\Big(\lambda^{-1}(x,w)\Big) = \lambda^{\alpha}\tau\Big((x,w)\Big)
	\end{aligned}
	\end{equation}
	Incorporating \eqref{theorem_proof_eq5} into \eqref{theorem_proof_eq4}, we finally get:
	\begin{equation}\label{theorem_proof_eq6}
	\begin{aligned}
	&\mu\Big((x,w),t\Big)\geq\sup\Big\{\phi\Big(\reach_t^{\tilde{F}}((x,0,w))\Big)\Big\},\\
	&\forall(x,w)\in \cone\text{ and }\forall t\in[0,\tau\Big((x,w)\Big)],
	\end{aligned}	
	\end{equation}
	i.e. $\mu(\cdot)$ satisfies \eqref{bound time validity} in $\cone$.
	
	Finally, let us prove that $\mu(\cdot)$ satisfies \eqref{one-zero-crossing of bound}. Observe that, since $\delta_0\geq 0$, $\delta_1>0$ and \eqref{delta_init_cond} holds, then $\mu\Big((x,w),t\Big)$ and $\dot{\mu}\Big((x,w),t\Big)$ are strictly increasing w.r.t. $t$ (for a more detailed proof, see \cite{delimpaltadakis_tac}). Thus, since $\mu\Big((x,w),0\Big)<0$, then, for any $(x,w)\in\real^{n+1}\setminus\{0\}$, $\exists!\tau^\downarrow(x,w)>0$ such that $\mu\Big((x,w),\tau^\downarrow\Big((x,w)\Big)\Big)=0$. The proof is now complete.
\end{proof}
\bibliography{bibliography_draft_stc_disturbances.bib}
\bibliographystyle{myIEEEtran} 

\end{document}